\documentclass[1p]{elsarticle}
\usepackage{amsmath}
\usepackage{pat}
\usepackage{url}

\newproof{proof}{Proof}

\newcommand{\eps}{\varepsilon}

\begin{document}
\title{Approximating Majority Depth\tnoteref{t1}}
\tnotetext[t1]{This research was supported by NSERC.}
\author[carleton]{Dan Chen and Pat Morin}

\address[carleton]{School of Computer Science, Carleton University, 1125 Colonel By Drive, Ottawa, CANADA, K1S~5B6}

\begin{abstract}
We consider the problem of approximating the majority depth (Liu and
Singh, 1993) of a point $q$ with respect to an $n$-point set, $S$,
by random sampling.  At the heart of this problem is a data structures
question: How can we preprocess a set of $n$ lines so that we can quickly
test whether a randomly selected vertex in the arrangement of these
lines is above or below the median level.  We describe a Monte-Carlo data
structure for this problem that can be constructed in $O(n\log n)$ time,
can answer queries in $O((\log n)^{4/3})$ expected time, and answers correctly
with high probability.
\end{abstract}

\maketitle

\section{Introduction}

A \emph{data depth measure} quantifies the centrality
of an individual (a point) with respect to a population (a point set).
Depth measures are an important part of multivariate statistics and many
have been defined, include Tukey depth \cite{t74}, Oja
depth \cite{o83}, simplicial depth \cite{l90}, majority depth \cite{ls93},
and zonoid depth \cite{dkm96}.  For an overview of data depth from a
statistician's point of view, refer to the survey by Small \cite{s90}.
For a computational geometer's point of view refer to Aloupis' survey
\cite{a06}.

In this paper, we focus on the bivariate majority depth measure.  Let $S$
be a set of $n$ points in $\R^2$.  For a pair $x,y\in S$, a \emph{major
side} of $x,y$ is a closed halfplane having both $x$ and $y$ on its
boundary and that contain at least $n/2$ points of $S$.  Note that,
in the degenerate case where the line through $x$ and $y$ is a halving
line, the pair $x,y$ has two major sides.  The \emph{majority depth}
\cite{ls93,s91} of a point, $q$, with respect to $S$, is defined as the
number of pairs $x,y\in S$ that have $q$ in a major side.

Under the usual projective duality \cite{e97}, the set $S$ becomes
a set, $S^*$, of lines; pairs of points in $S$ becomes vertices
in the arrangement, $A(S^*)$, of $S^*$; and $q$ becomes a line, $q^*$.
The \emph{median-level} of $A(S^*)$ is the closure of the set of points
on lines in $S$ that have exactly $\lfloor n/2\rfloor$ lines of $S$
above them.  Then the majority depth of $q$ with respect to $S$ is equal
to the number of vertices, $x$, in $A(S^*)$ such that
\begin{enumerate}
\item $x$ is above $q^*$ and $x$ is above the median level; or
\item $x$ is below $q^*$ and $x$ is below the median level.
\end{enumerate}

Chen and Morin \cite{cm11} present an algorithm for computing majority
depth that works in the dual.  Their algorithm works by computing the
median level, computing the intersections of $q^*$ with the median
level, and using fast inversion counting to determine the number, $t$,
of vertices of the arrangement sandwiched between $q^*$ and the median
level.  The majority depth of $q$ is then equal to $\binom{n}{2}-t$.
The running time of this algorithm is within a logarithmic factor of $m$,
the complexity of the median level.

The maximum complexity of the median level of $n$ lines has been the
subject of intense study since it was first posed.  The current best
upper bound is $O(n^{4/3})$, due to Dey \cite{d98} and the current best
lower bound is $n2^{\Omega(\sqrt{\log n})}$, due to T\'oth \cite{t00}
and tightened by Nivasch \cite{n08}.
The median level can be computed in time $O(\min\{m\log n,n^{4/3}\})$
\cite{bj02,c99}.  Thus, the worst-case running time of Chen and Morin's
majority depth algorithm is $\omega(n(\log n)^c)$ for any constant $c$,
but no worse than $O(n^{4/3}\log n)$.

It seems difficult for any algorithm that computes the exact majority
depth of a point to avoid (at least implicitly) computing the median
level of $A(S^*)$.  This motivates approximation by random sampling.
In particular, one can use the simple technique of sampling vertices of
$A(S^*)$ and checking, for each sampled vertex, whether
\begin{enumerate}
  \item it lies above or below $q^*$; and
  \item it lies above or below the median level of $S^*$.
\end{enumerate}
In the primal, this is equivalent to taking random pairs of points in
$S$ and checking, for each such pair, $(x,y)$, if, (1)~the closed
upper halfplane, $h_{xy}$, with $x$ and $y$ on its boundary, contains $q$
and (2)~if $h_{xy}$ contains
 $n/2$ or more points of $S$.

The former test takes constant time but the latter test leads to a
data structuring problem:  Preprocess the set $S^*$ so that one can
quickly test, for any query point, $x$, whether $x$ is above or below
the median level of $A(S^*)$.  (Equivalently, does a query halfplane, $h$,
contain $n/2$ or more points of $S$.)  We know of two immediate solutions
to this problem.  The first solution is to compute the median level
explicitly, in $O(\min\{m\log n,n^{4/3}\})$ time, after which any query
can be answered in $O(\log n)$ time by binary search on the x-coordinate
of $x$.  The second solution is to construct a half-space range counting
structure---a partition tree---in $O(n\log n)$ time that can count the
number of points of $S$ in $h_{xy}$ in $O(n^{1/2})$ time \cite{c12}.

The first solution is not terribly good, since Chen and Morin's algorithm
shows that computing the \emph{exact} majority depth of $q$ can be done
in time that is within a logarithmic factor of $m$, the complexity of
the median level.  (Though if the goal is to preprocess in order to
approximate the majority depth for many different points, then this
method may be the right choice.)

In this paper, we show that the second solution can be improved
considerably, at least when the application is approximating majority
depth.  In particular, we show that when the query point is a randomly
chosen vertex of the arrangement $A(S^*)$, a careful combination of
partition trees \cite{c12} and $\eps$-approximations \cite{mww93}
can be used to answer queries in $O((\log n)^{4/3})$ expected time.
This faster query time means that we can use more random samples which
leads to a more accurate approximation.

The remainder of this paper is organized as follows.  In
\secref{range-counting} we review results on range counting and
$\eps$-approximations and show how they can be used for approximate range
counting.  In \secref{fast-testing} we show how these approximate range
counting results can be used to quickly answer queries about whether
a random vertex of $S^*$ is above or below the median level of $S^*$.
In \secref{majority-depth} we briefly mention how all of this applies to
the problem of approximating majority depth.  Finally, \secref{conclusion}
concludes with an open problem.

\section{Approximate Range Counting}
\seclabel{range-counting}

In this section, we consider the problem of approximate range counting.
That is, we study algorithms to preprocess $S$ so that, given a closed
halfplane $h$ and an integer $i\ge 0$, we can quickly return an integer
$r_i(h,S)$ such that
\[
   \left| |h\cap S| - r_i(h,S)\right| \le i \enspace .
\]
This data structure is such that queries are faster when the allowable
error, $i$, is larger.

There are no new results in this section. Rather it is a review of two
relevant results on range searching and $\epsilon$-approximations that are
closely related, but separated by nearly 20 years.  The reason we do this
is that, without a guide, it can take some effort to gather and assemble
the pieces; some of the proofs are existential, some are stated in terms
of discrepancy theory, and some are stated in terms of VC-dimension.
The reader who already knows all this, or is uninterested in learning it,
should skip directly to \lemref{approx-range-counting}.

The first tools we need come from a recent result of Chan on optimal
partition trees and their application to exact halfspace range counting
\cite[Theorems 3.2 and 5.3, with help from Theorem~5.2]{c12}:

\begin{thm}\thmlabel{optimal-partition-tree}
  Let $S$ be a set of $n$ points in $\R^2$ and let $N\ge n$ be an integer.
  There exists a data structure that can preprocess $S$ in $O(n\log
  N)$ expected time so that, with probability at least $1-1/N$, for
  any query halfplane, $h$, the data structure can return $|h\cap S|$
  in $O(n^{1/2})$ time.
\end{thm}

We say that a halfplane, $h$, \emph{crosses} a set, $X$, of points if
$h$ neither contains $X$ nor is disjoint from $X$.  The partition tree
of \thmref{optimal-partition-tree} is actually a \emph{binary space
partition tree}.  Each internal node, $u$, is a subset of $\R^2$ and
the two children of a node are the subsets of $u$ obtained by cutting
$u$ with a line.  Each leaf, $w$, in this tree has $|w\cap S| \le 1$.
The $O(n^{1/2})$ query time is obtained by designing this tree so that,
with probability at least $1-1/N$, there are only $O(n^{1/2})$ nodes
crossed by any halfplane.

For a geometric graph $G=(S,E)$, the \emph{crossing number} of $G$ is the
maximum, over all halfplanes, $h$, of the number of edges $uw\in E$ such
that $h$ crosses the set $\{u,w\}$.  From \thmref{optimal-partition-tree} it is
easy to derive a spanning tree of $S$ with crossing number $O(n^{1/2})$
using a bottom-up algorithm:  Perform a post-order traversal of the
partition tree.  When processing a node $u$ with children $v$ and $w$,
add an edge to the tree that joins an arbitrary point in $v\cap S$ to
an arbitrary point in $w\cap S$.  Since a halfplane cannot cross any
edge unless it also crosses the node at which the edge was created,
this yields the following result \cite[Corollary~7.1]{c12}:

\begin{thm}\thmlabel{spanning-tree}
  For any $n$ point set, $S$, and any $N\ge n$, it is possible to compute,
  in $O(n\log N)$ expected time, a spanning tree, $T$, of $S$ that,
  with probability at least $1-1/N$, has crossing number $O(n^{1/2})$.
\end{thm}

A spanning tree is not quite what is needed for what follows.  Rather,
we require a matching of size $\lfloor n/2\rfloor$.  To obtain this,
we first convert the tree, $T$, from \thmref{spanning-tree} into a path
by creating a path, $P$, that contains the vertices of $T$ in the order
they are discovered by a depth-first traversal.  It is easy to verify that
the crossing number of $P$ is at most twice the crossing number of $T$.
Next, we take every second edge of $P$ to obtain a matching:

\begin{cor}\corlabel{matching}
  For any $n$ point set, $S$, and any $N\ge n$, it is possible to
  compute, in $O(n\log N)$ expected time, a matching, $M$, of $S$ of
  size $\lfloor n/2\rfloor$ that, with probability at least $1-1/N$
  has crossing number $O(n^{1/2})$.
\end{cor}

The following argument is due to Matou\v{s}ek, Welzl and Wernsich
\cite[Lemma~2.5]{mww93}.  Assume, for simplicity, that $n$ is even and let
$S'\subset S$ be obtained by taking exactly one endpoint from each edge in
the matching $M$ obtained from \Corref{matching}.  Consider some halfplane
$h$,  and let $M^{I}_h$ be the subset of the edges of $M$ contained in
$h$ and let $M^{C}_h$ be the subset of edges crossed by $h$. Then
\[
     |h\cap S| = 2|M^{I}_h| + |M^{C}_h| \enspace .
\] 
In particular,
\[
     |h\cap S| - |M^{C}_h| \le 2|h\cap S'| \le |h\cap S| + |M^{C}_h| 
\]
Since $|M^C_h|\in O(n^{1/2})$, this is good news in terms of approximate
range counting;  the set $S'$ is half the size of $S$, but $2|h\cap
S'|$ gives estimate of $|h\cap S|$ that is off by only $O(n^{1/2})$.
Next we show that this can be improved considerably with almost no effort.

Rather than choosing an arbitrary endpoint of each edge in $M$ to
take part in $S'$, we choose each one of the two endpoints at random
(and independently of the other $n/2-1$ choices).  Then, each edge in
$M^{C}_h$ has probability $1/2$ of contributing a point to $h\cap S'$
and each edge in $M^{I}_h$ contributes exactly one point to $h\cap S'$.
Therefore,
\[
    \E[2|h\cap S'|]
      = 2\left(1|M^{I}_h| + \frac{1}{2}|M^{C}_h|\right) = |h\cap S| \enspace .
\]
That is, $2|h\cap S'|$ is an unbiased estimator of $|h\cap
S|$.  Even better: the error of this estimator is (2 times) a
binomial$(|M^{C}_h|,1/2)$ random variable, with $|M^{C}_h|\in O(n^{1/2})$.
Using standard results on the concentration of binomial random variables
(i.e., Chernoff Bounds \cite{c52}), we immediately obtain:
\[
   \Pr\{\left|2|h \cap S'| - |h\cap S|\right| \ge c n^{1/4}(\log N)^{1/2}\} 
       \le 1/N \enspace ,
\]
for some constant $c>0$.  That is, with probability $1-1/N$, $2|h\cap S'|$
estimates $|h\cap S|$ to within an error of $O(n^{1/4}(\log N)^{1/2})$.
Putting everything together, we obtain:

\begin{lem}\lemlabel{yeah}
  For any $n$ point set, $S$, and any $N\ge n$, it is possible to
  compute, in $O(n\log N)$ expected time, a subset $S'$ of $S$ of 
  size $\lceil n/2\rceil$ such that, 
  with probability at least $1-1/N$, for every halfplane $h$,
  \[
     \left|2|h\cap S'| - |h\cap S|\right| \in O(n^{1/4}(\log N)^{1/2}) \enspace .
  \]
\end{lem}

What follows is another argument by Matou\v{s}ek, Welzl and Wernisch
\cite[Lemma~2.2]{mww93}.  By repeatedly applying \lemref{yeah}, we
obtain a sequence of $O(\log n)$ sets $S_0\supset S_1\cdots\supset S_r$,
$S_0=S$ and $|S_j|=\lceil n/2^j\rceil$.  For $j\ge 1$, the set $S_j$
can be computed from $S_{j-1}$ in $O(2^{-j}n\log N)$ time and has the
property that, with probability at least $1-1/N$, for every halfplane $h$,
\begin{equation}
   \left|2^j|h\cap S_j| - |h\cap S|\right| \in O(2^{3j/4}n^{1/4}(\log N)^{1/2}) \enspace .
  \eqlabel{crap}
\end{equation}
At this point, we have come full circle.  We store each
of the sets $S_0,\ldots,S_r$ in an optimal partition tree
(\thmref{optimal-partition-tree}) so that we can do range counting
queries on each set $S_i$ in $O(|S_i|^{1/2})$ time.  This (finally)
gives the result we need on approximate range counting:
\begin{lem}\lemlabel{approx-range-counting}
  Given any set $S$ of $n$ points in $\R^2$ and any $N\ge n$, there exists
  a data structure that can be constructed in $O(n\log N)$ expected time
  and, with probability at least $1-1/N$, can, for any halfspace $h$
  and any integer $i\in\{0,\ldots,n\}$, return a number $r_i(h,S)$ such that
  \[  \left||h\cap S|-r_i(h,S)\right| \le i \enspace .\]
  Such a query takes $O(\min\{n^{1/2},(n/i)^{2/3}(\log N)^{1/3}\})$ expected time.
\end{lem}

\begin{proof}
  The data structure is a sequence of optimal partition trees on the
  sets $S_0,\ldots,S_r$.  All of these structures can be computed in
  $O(n\log N)$ time, since $|S_0|=n$ and the size of each subsequent
  set decreases by a factor of 2.

  To answer a query, $(h,i)$, we proceed as follows: If $i\le n^{1/4}$,
  then we perform exact range counting on the set $S_0=S$ in $O(n^{1/2})$
  time to return the value $|h\cap S|$.  Otherwise, we perform range
  counting on the set $S_j$ where $j$ is the largest value that satisfies
  \[
      C2^{3j/4} n^{1/4}(\log N)^{1/2} \le i \enspace ,
  \]
  where the constant $C$ depends on the constant in the big-Oh notation
  in \eqref{crap}.  This means that $|S_j| = O((n/i)^{4/3}(\log N)^{2/3}))$
  and the query takes expected time
  \[
      O(|S_j|^{1/2}) = O((n/i)^{2/3}(\log N)^{1/3}) \enspace ,
  \]
   as required. \qed
\end{proof}

Our main application of \lemref{approx-range-counting} is a test that
checks whether a halfspace, $h$, contains $n/2$ or more points of $S$.

\begin{lem}\lemlabel{side-test}
  Given any set $S$ of $n$ points in $\R^2$ and any $N\ge n$, there
  exists a data structure that can be constructed in $O(n\log N)$
  expected time and, with probability at least $1-1/N$, can, for any
  halfspace $h$, determine whether $|h\cap S|\ge n/2$ or not.  Such a
  query takes expected time
  \[  Q(i) = \begin{cases}
        O(n^{1/2}) & \text{for $0\le i\le n^{1/4}$} \\
        O((n/i)^{2/3}(\log N)^{1/3}) & \text{otherwise, }
    \end{cases}
  \]
  where $i = ||h\cap S|-n/2|$.
\end{lem}

\begin{proof}
  As preprocessing, we construct the data structure of
  \lemref{approx-range-counting}.  To perform a query, we perform a
  sequence of queries $(h,i_j)$, for $j=0,1,2,\ldots$, where $i_j=n/2^j$.
  The $j$th such query takes $O(2^{2j/3}(\log N)^{1/3})$ time and the
  question, ``is $|h\cap S|\ge n/2$?'' is resolved once $n/2^j < i/2$.
  Since the cost of successive queries is exponentially increasing, this
  final query takes time $O(\min\{n^{1/2},(n/i)^{2/3}(\log N)^{1/3}\})$
  and dominates the total query time. \qed
\end{proof}

\section{Side of Median Level Testing}
\seclabel{fast-testing}

We are now ready to tackle the main problem that comes up in trying
to estimate majority depth by random sampling:  Given a range pair of
points $x,y\in S$, determine if there are more than $n/2$ points in
the upper halfspace, $h_{xy}$, whose boundary is the line through $x$
and $y$.  In this section, though, it will be more natural to work in
the dual setting.  Here the question becomes: Given a random vertex,
$x$, of $A(S^*)$, determine whether $x$ is above or below the median
level of $S^*$.  The data structure in \lemref{side-test} answers these
queries in time $O((n/i)^{2/3}(\log N)^{1/3})$ when the vertex $x$ is on the
$n/2-i$ or $n/2+i$ level.

Before proving our main theorem, we recall a result of Dey
\cite[Theorem~4.2]{d98} about the maximum complexity of a sequence
of levels.

\begin{lem}\lemlabel{dey}
 Let $L$ be any set of $n$ lines and let $s$ be the number of vertices
 of $A(L)$ that are on levels $k$, $k+1$,\ldots, or $k+j$.  Then, $s \in
 O(nk^{1/3}j^{2/3})$.
\end{lem}

We are interested in the special case of \lemref{dey} where $k=n/2-i$
and $j=2i$:

\begin{cor}\corlabel{dey}
  Let $L$ be any set of $n$ lines.  Then, for any $i\in\{1,\ldots,n/2\}$
  the maximum total number of vertices of $A(L)$ whose level is in
  $\{n/2-i,\ldots,n/2+i\}$ is $O(n^{4/3}i^{2/3})$.
\end{cor}

\Corref{dey} is useful because it gives bounds on the distribution of
the level of a randomly chosen vertex of $A(S^*)$.


\begin{thm}\thmlabel{exp-side-test}
  Given any set, $L$, of $n$ lines and any $c>0$, there exists a data
  structure that can test if a point $x$ is above or below the median
  level of $L$.  For any constant, $c$, this structure can be made to
  have the following properties:
  \begin{enumerate}\setlength{\itemsep}{0mm}
    \item It can be constructed in
       $O(n\log n)$ expected time and uses $O(n)$ space;
    \item with probability
       at least $1-n^{-c}$, it answers correctly for all
       possible queries; and
    \item when given a random vertex of $A(L)$
       as a query, the expected query time is $O((\log n)^{4/3})$.
  \end{enumerate}
\end{thm}

\begin{proof}
  The data structure is, of course, the data structure of
  \lemref{side-test} with $N=n^c$.    If the query vertex is chosen from
  level $n/2-i$ or $n/2-i$, the query time is upper-bounded by
  \[
        Q(i) \le \begin{cases}
          O(n^{1/2}) & \text{if $0\le i \le n^{1/4}$} \\
          O((n/i)^{2/3}(\log N)^{1/3})  & \text{otherwise} .
        \end{cases}
  \]
  For simplicity, assume that $n$
  is even.
  Let $n_i$ be the total number of vertices of $A(L)$ on levels $n/2-i+1$
  and $n/2+i-1$, with the convention that $n_i=0$ for all $i>n/2$.
  Then the expected query time of this data structure is at most
  \begin{eqnarray*}
    T
      & \le & \frac{1}{\binom{n}{2}}\sum_{i=1}^{n/2} n_i Q(i) \\
      & = & \frac{1}{\binom{n}{2}}\sum_{k=0}^{\lceil\log n\rceil}\left(\sum_{i=2^k}^{2^{k+1}-1} n_i Q(i)\right) \\
      & \le & \frac{1}{\binom{n}{2}}\sum_{k=0}^{\lceil\log n\rceil}O(n^{4/3}2^{2k/3})\cdot Q(2^{k}) \\
      & = & \frac{1}{\binom{n}{2}}\left(
              \sum_{k=0}^{\lfloor\log n^{1/4}\rfloor}O(n^{4/3}2^{2k/3})\cdot Q(2^{k}) 
              + \sum_{k=\lfloor\log n^{1/4}\rfloor+1}^{\lceil\log n\rceil}O(n^{4/3}2^{2k/3})\cdot Q(2^{k})\right) \\
      & = & \frac{1}{\binom{n}{2}}\left(O(n^{4/3+(1/4)2/3})  
              + \sum_{k=\lfloor\log n^{1/4}\rfloor+1}^{\lceil\log n\rceil}O(n^{4/3}2^{2k/3})\cdot O((n/2^k)^{2/3}(\log N)^{1/3}) \right) \\
      & = & \frac{1}{\binom{n}{2}}\left(O(n^{4/3+(1/4)2/3})  
              + \sum_{k=\log n^{1/4}+1}^{\log n}O(n^2(\log N)^{1/3}) \right) \\
      & = & O((\log N)^{4/3}) = O((\log n)^{4/3}) \enspace ,
  \end{eqnarray*}
  as required. \qed
\end{proof}

\section{Estimating Majority Depth}
\seclabel{majority-depth}

Finally, we return to our application, namely estimating majority depth.
For a set $S$ of $n$ points in $\subset\R^2$ and a point $q\in\R^2$,
let $d(q,S)$ denote the majority depth of $q$ with respect to $S$ and
let $p=d(q,S)/\binom{n}{2}$ denote the normalized majority depth of $q$.


\begin{thm}
  Given a set $S$ of $n$ points in $\R^2$ and any constant $c>0$, there
  exists a data structure that can preprocess $S$ using $O(n\log n)$
  expected time and $O(n)$ space, such that, for any point $q$ and any
  parameter $r\ge 1$, the data structure can compute, in $O(r(\log n)^{4/3})$
  expected time, a value $d'(q,S)$ such that
  \[
     \Pr\left\{\frac{|d'(q,S)-d(q,S)|}{d(q,S)} \ge \eps \right\} 
        \le \exp\left(-\Omega\left(\eps^2rp\right)\right) + n^{-c} \enspace .
  \]
  In particular, choosing $r=(c\log n)/(\eps^2p)$ yields a linear-sized
  data structure that, with probability at least $1-n^{-c}$, provides an
  $\eps$-approximation to $d(q,S)$ in $O(c(\log n)^{7/3}/(\eps^2p))$ time.
\end{thm}

\begin{proof}
  The data structure is the one described in \thmref{exp-side-test}.
  Select $r$ random vertices of $A(S^*)$ (by taking random pairs of lines
  in $S^*$) and, for each sample, test if it contributes to $d(q,S)$.
  This yields a count $r' \le r$ where
  \[ 
     \E[r'] = rp \enspace .
  \]
  We then return the value $d'(q,S)=(r'/r)\binom{n}{2}$, so that
  $\E[d'(q,S)]=d(q,S)$, as required.

  To prove the error bound, we use the fact that $r'$ is a binomial$(p,r)$
  random variable.  Applying Chernoff Bounds \cite{c52} on $r'$ yields:
  \[
     \Pr\{|r' - rp| \ge \eps rp\} \le \exp(-\Omega(\eps^2rp)) \enspace .
  \]
  Finally, the algorithm may fail not because of badly chosen samples,
  but rather, because the data structure of \thmref{exp-side-test} fails.
  The probability that this happens is at most $n^{-c}$. Therefore,
  the overall result follows from the union bound. \qed
\end{proof}

\section{Conclusions}
\seclabel{conclusion}

Although the estimation of majority depth is the original motivation for
studying this problem, the underlying question of the tradeoffs involved
in preprocessing for testing whether a point is above or below the median
level seems to be a fundamental question that is still far from being answered.
In particular, we have no good answer to the following question:

\begin{op}
What is the fastest linear-space data structure for testing if an
arbitrary query point is above or below the median level of a set of
$n$ lines?
\end{op}

To the best of our knowledge, the current state of the art is partition
trees, which can only answer these queries in $O(n^{1/2})$ time.

\section*{Acknowledgement}

In the preliminary version of this paper \cite{cm12}, the proof of
\thmref{exp-side-test} involved a much longer calculation.  The authors
are grateful to an anonymous referee who pointed out the trick of grouping
into powers of two, which simplifies the proof into its current form.

\section*{References}
\bibliographystyle{elsarticle-num-names}
\bibliography{majapx}

\end{document}